\documentclass[11pt]{amsart}
\usepackage{etoolbox}
\usepackage[margin=.89in]{geometry}

\usepackage{parskip}
\usepackage{amssymb, amsmath, amsthm}
\usepackage{breqn}
\usepackage{verbatim}
\usepackage{enumerate}
\usepackage{fancyhdr}
\usepackage{graphicx}
\usepackage{varwidth}
\usepackage{bm}
\usepackage{authblk}
\usepackage[thinlines]{easytable}

\usepackage{hyperref}

\usepackage{algpseudocode}
\usepackage{float}
\floatstyle{boxed}
\restylefloat{figure}
\restylefloat{table}

\usepackage[symbol]{footmisc}

\theoremstyle{plain}

\newtheorem{theorem}{Theorem}[section]
\newtheorem{lemma}[theorem]{Lemma}

\theoremstyle{definition}

\newtheorem{problem}[theorem]{Problem}

\theoremstyle{remark}
\newtheorem{remark}[theorem]{Remark}

\newcommand{\mb}[1]{\ensuremath{\mathbb{#1}}}

\newcommand{\Z}{\mb{Z}}

\newcommand{\F}{\mb{F}}

\usepackage[foot]{amsaddr}

\usepackage{algorithm}
\usepackage{algpseudocode}
\algnewcommand\algorithmicinput{\textbf{Input:}}
\algnewcommand\INPUT{\item[\algorithmicinput]}
\algnewcommand\algorithmicoutput{\textbf{Output:}}
\algnewcommand\OUTPUT{\item[\algorithmicoutput]}
\algnewcommand\algorithmicoracle{\textbf{Oracle:}}
\algnewcommand\ORACLE{\item[\algorithmicoracle]}
\makeatletter
\algnewcommand{\LineComment}[1]{\Statex \hskip\ALG@thistlm \(\triangleright\) #1}
\makeatother

\usepackage{tikz}
\usetikzlibrary{matrix,arrows}

\begin{document}

\title{Fast Computation of Isomorphisms Between Finite Fields Using Elliptic Curves}

	\author{Anand Kumar Narayanan\textsuperscript{1}}
	\address{\textsuperscript{1} Laboratoire d'Informatique de Paris 6, Sorbonne Universite, Paris.}
	\footnotetext{Supported by NSF grant \#CCF-1423544 and European Union's H2020 Programme (grant agreement \#ERC-669891).}
	
	\email{anand.narayanan@lip6.fr}
	\setcounter{Maxaffil}{0}
	\renewcommand\Affilfont{\itshape\small}
	\maketitle
	\vspace{-10pt}
	
%	\footnotetext{\textsuperscript{1} Supported by NSF grant \#CCF-1423544, Chris Umans' Simons Foundation Investigator grant and European Union's H2020 Programme under grant agreement number ERC-669891.}

\begin{abstract}
We propose a randomized algorithm to compute isomorphisms between finite fields using elliptic curves. To compute an isomorphism between two fields of cardinality $q^n$, our algorithm takes $$n^{1+o(1)} \log^{1+o(1)}q + \max_{\ell} \left(\ell^{n_\ell + 1+o(1)} \log^{2+o(1)} q + O(\ell \log^5q)\right)$$
%$$\widetilde{O}\left(n \log^2 q + \max_{\ell} (\ell^{n_\ell + 1} \log^2 q + \ell \log^5q)\right)$$
time, where $\ell$ runs through primes dividing $n$ but not $q(q-1)$ and $n_\ell$ denotes the highest power of $\ell$ dividing $n$.
Prior to this work, the best known run time dependence on $n$ was  quadratic. Our run time dependence on $n$ is at worst quadratic but is subquadratic if $n$ has no large prime factor. 
%Prior to this work, the best known running time exponent on $n$ for computing a finite field isomorphism was nearly quadratic. Our exponent is never worse and is significantly smaller unless $n$ has a large prime factor. % Our exponent is never worse and is often (unless $n$ has a large prime factor) significantly smaller.
In particular, the $n$ for which our run time is nearly linear in $n$ have natural density at least $3/10$.
% For instance, $n$ for which our run time is nearly linear in $n$ has natural density at least $3/10$. %Our running time is no worse and for several natural infinite families of $(q,n)$ nearly linear in $n$. 
%For several families of $(q,n)$, such as when $n$ is $\mathcal{O}(\sqrt{n})$ smooth or when $n$ is supported above primes dividing $q(q-1)$, our running time is nearly linear in $n$. % The case when $n$ is $\mathcal{O}(\sqrt{n})$ smooth has positive density in the integers.
%Prior to this work, the best known running time dependence on $n$ was nearly quadratic. 
%Our approach is to reduce computing isomorphisms to computing discrete logarithms in certain small torsion subgroups of elliptic curves over finite fields. The crux of the reduction is finding preimages under the Lang map.  
The crux of our approach is finding a point on an elliptic curve of a prescribed prime power order or equivalently finding preimages under the Lang map on elliptic curves over finite fields. We formulate this as an open problem whose resolution would solve the finite field isomorphism problem with run time nearly linear in $n$.
% between finite fields.%$(n^{1+o(1)})(\log q)^{O(1)}$ time algorithm.  \\ \\
%We describe a randomized algorithm that given a polynomial over a finite field $\F_q$ and a positive integer $n$, finds a root in $\F_{q^n}$.
%We describe a randomized algorithm that given a degree $m$ polynomial over a finite field $\F_q$ and a positive integer $n$, finds a root in $\F_{q^n}$. Construction of the field extension $\F_{q^n}$ could either be left to the algorithm or given explicitly as an input. The latter allows us to compute an embedding of a finite field in another. In particular, when presented with two explicitly presented finite fields of the same cardinality, it allows us to compute an isomorphim. \\ \\
%The expected running time of our algorithm is $(m^{1+o(1)}+n^{1+o(1)}+ \max_\ell \ell^{a_\ell+1}) (\log q)^{O(1)}$ where the $\max$ is taken over primes $\ell$ dividing $n$ but not $q(q-1)$ and $a_\ell$ is the largest power of $\ell$ dividing $n$. Prior to this work, the best known running time was $(mn)^{1+o(1)} (\log q)^{1+o(1)}$. Since $m$ may be assumed to be at least $n$ without loss of generality, our running time is no worse. Except when $n$ has a large prime factor, our running time is significantly faster.\\ \\ 
%The crux of our approach is finding preimages under the Lang map on elliptic curves over finite fields. We formulate a computational problem concerning Lang maps whose resolution would yield a $(m^{1+o(1)}+n^{1+o(1)})(\log q)^{O(1)}$ time algorithm for our root finding problem.  \\ \\
\end{abstract}
%\part{Use this type of header for very long papers only}
% use lowercase except for proper names

\section{Introduction}
\subsection{Computing Isomorphisms between Finite Fields}
Every finite field has prime power cardinality, for every prime power there is a finite field of that cardinality and every two finite fields of the same cardinality are isomorphic. This now well known result due to Moore \cite{moo} poses two algorithmic problems. The first concerns field construction: given a prime power, construct a finite field of that cardinality. The second is the isomorphism problem: compute an isomorphism between two explicitly presented finite fields of the same cardinality.\\ \\
Field construction is performed by constructing an irreducible polynomial of appropriate degree over the underlying prime order field, with all known efficient unconditional constructions requiring randomness. The fastest known construction, due to Couveignes and Lercier \cite{cl} uses elliptic curve isogenies. In practice, a polynomial is chosen at random and tested for irreducibility \cite{ben}. Such non canonical construction of finite fields motivates the isomorphism problem in several applications. For instance in cryptography, the discrete logarithm problem over small characteristic finite fields is often posed over fields constructed using random irreducible polynomials. In cryptanalysis, the quasi-polynomial algorithm \cite{bgjt} for discrete logarithms works over fields constructed using irreducible polynomials of a special form. An isomorphism computation is thus  required as a preprocessing step in cryptanalysis.\\ \\% No deterministic polynomial time algorithms for field construction are known.\\ \\
%Remarkably, the isomorphism problem was shown to be in deterministic polynomial time by Lenstra \cite{len}. There are faster randomized algorithms. 
Zierler noted that the isomorphism problem reduces to root finding over finite fields and hence has efficient randomized algorithms \cite{zie}. Remarkably, the isomorphism problem was shown to be in deterministic polynomial time by Lenstra \cite{len}. Allombert \cite{all} proposed a linear algebraic randomized algorithm, close in spirit with Lenstra's algorithm but markedly faster. Employing the (randomized) polynomial factorization algorithm of Kaltofen and Shoup \cite{ks} (implemented using Kedlaya-Umans fast modular composition \cite{ku}) to find roots, Zierler's approach yields the fastest previously known algorithm for computing isomorphisms. Our main result is an algorithm with improved run time in most cases.\\ \\
An alternate approach relying on cyclotomy instead of root finding was introduced by Pinch \cite{pin} and improved upon by Rains \cite{rai} to give the fastest algorithm in practice. The cyclotomic method of Pinch requires that the finite fields in question contain certain small order roots of unity. To remove this requirement, Pinch \cite{pin} proposed using elliptic curves over finite fields. This way, instead  of roots of unity, one seeks rational points of small order on elliptic curves. Our algorithm, although very different, relies on elliptic curves as well. We take inspiration from the aforementioned algorithm of Couveniges and Lercier \cite{cl}. While \cite{cl} used elliptic curve isogenies to solve field construction in nearly linear time, we solve the isomorphism problem. Our algorithm may also be viewed as an extension of Allombert's \cite{all} using elliptic curves.\\ \\
A critical component of our approach is a method to reduce the isomorphism problem for arbitrary degrees to prime power degrees in nearly linear time. The reduction is mostly subtle linear algebra and similar to a theorem of Shoup \cite{sho_ir}[Theorem 5]\footnote{We thank an anonymous referee for pointing out the similarity}. We invoke elliptic curves only to solve the prime power cases. \\ \\
Soon after posting a preprint version of the current paper online \cite{nar}, I was informed of concurrent related work that later appeared here \cite{bddfs}. Therein Brieulle, De Feo, Doliskani, Flori and Schost address the very same isomorphism problem (and more generally finite field embedding problems) using elliptic curve based techniques similar to ours. In contrast to our emphasis on establishing complexity theoretic bounds on the isomorphism problem, their goals are directed towards obtaining fast practical algorithms. In particular, they present an open source implementation of their algorithm which appears to be the current state of the art in practice.
\subsection{Computing Isomorphisms and Root Finding}
We formally pose the isomorphism problem stating the manner in which the input fields and the output isomorphism are represented. 
Let $q$ be a power of a prime $p$ and let $\F_q$ denote the finite field with $q$ elements. Fix an algebraic closure $\overline{\F}_q$ of $\F_q$ and let $\sigma:\overline{\F}_q \longrightarrow \overline{\F}_q$ denote the $q^{th}$ power Frobenius endomorphism.  We consider two finite fields of cardinality $q^n$ to be given through two monic irreducible degree $n$ polynomials $f(x),g(x) \in \F_q[x]$. The fields are then constructed as $\F_q(\alpha)$ and $\F_q(\beta)$ where $\alpha,\beta \in \overline{\F}_q$ are respectively roots of $f(x),g(x)$. Without loss of generality \cite{cl}, all our algorithms assume the base field $\F_q$ to be given as the quotient of the polynomial ring over $\Z/p\Z$ by a monic irreducible polynomial over $\Z/p\Z$.\\ \\
An isomorphism $\phi: \F_q(\alpha) \longrightarrow \F_q(\beta)$ that fixes $\F_q$ is completely determined by the image $\phi(\alpha)$. We call the unique $r_\phi(x) \in \F_q[x]$ of degree less than $n$ such that $\phi(\alpha)=r_\phi(\beta)$ as the polynomial representation of $\phi$. We are justified in seeking the polynomial representation of $\phi$ since given $r_\phi(x)$, one may compute the image of an element in $\F_q(\alpha)$ under $\phi$ in time nearly linear in $n$ using fast modular composition \cite{ku}. For an $r(x) \in \F_q[x]$ of degree less than $n$, $r(x)$ is the polynomial representation of an isomorphism from $\F_q(\alpha)$ to $\F_q(\beta)$ if and only if $r(\beta)$ is a root of $f(x)$. Hence the problem of computing the polynomial representation of an isomorphism that fixes $\F_q$ is identical to the following root finding problem.\\ \\
%An isomorphism $\phi: \F_q(\beta) \longrightarrow \F_q(\alpha)$ that fixes $\F_q$ is completely determined by the image $\phi(\beta)$. We call the unique $r_\phi(x) \in \F_q[x]$ of degree less than $n$ such that $\phi(\beta)=r_\phi(\alpha)$ as the polynomial representation of $\phi$. We are justified in seeking the polynomial representation of $\phi$ since given $r_\phi(x)$, one may compute the image of an element in $\F_q(\beta)$ under $\phi$ in time nearly linear in $n$ using fast modular composition \cite{ku}. For an $r(x) \in \F_q[x]$ of degree less than $n$, $r(x)$ is the polynomial representation of an isomorphism from $\F_q(\beta)$ to $\F_q(\alpha)$ if and only if $r(\beta)$ is a root of $f(x)$. Hence the problem of computing the polynomial representation of an isomorphism that fixes $\F_q$ is identical to the following root finding problem.\\ \\  
% \ref{root_finding_problem}.
% Unlike \cite{len}, where isomorphisms are computed in matrix form, we look for the  polynomial representation.
%\begin{isomorphism}\label{root_finding_problem}
%Given monic irreducibles $f(x),g(x) \in \F_q[x]$ of degree $n$, find a root of $f(x)$ in $\F_q(\beta)$ where $\beta \in \overline{\F}_q$ is a root of $g(x)$.
%\end{isomorphism}
\textsc{Isomorphism Problem}: Given monic irreducibles $f(x),g(x) \in \F_q[x]$ of degree $n$, find a root of $f(x)$ in $\F_q(\beta)$ where $\beta \in \overline{\F}_q$ is a root of $g(x)$.\\ \\
%\begin{problem}\label{root_finding_problem}
%Given monic irreducibles $f(x),g(x) \in \F_q[x]$ of degree $n$, find a root of $f(x)$ in $\F_q(\beta)$ where $\beta \in \overline{\F}_q$ is a root of $g(x)$.
%\end{problem}
There are two input size parameters, namely $n$ and $\log q$. Prior to our work, the best known run time was quadratic in $n$ resulting from using \cite{ks,ku} to find roots in the \textsc{Isomorphism Problem}. We are primarily interested in lowering the run time exponent in $n$. Our run time dependence on $\log q$ will be polynomial but not optimized for. Here on, all our algorithms are Las Vegas randomized and by run time we mean expected run time. 
%Further, run times are stated using soft $\widetilde{O}$ notation that suppresses $n^{o(1)}$ and $\log^{o(1)} q$ terms. 
%\noindent An algorithm of Kaltofen and Shoup \cite{ks_h}, designed to factor polynomials over large degree extensions, implemented using the fast modular composition of Kedlaya and Umans \cite{ku} solves  problem \ref{root_finding_problem} in $n^{2+o(1)} (\log q)^{1+o(1)} + n^{1+o(1)} (\log q)^{2+o(1)}$ expected time.
\subsection{Summary of Results}
We present an algorithm for the \textsc{Isomorphism Problem} with run time
$$n^{1+o(1)} \log^{1+o(1)}q + \max_{\ell} \left(\ell^{n_\ell + 1+o(1)} \log^{2+o(1)} q + O(\ell \log^5q)\right)$$
% $$\widetilde{O}\left(n \log^2 q + \max_{\ell} (\ell^{n_\ell + 1} \log^2 q + \ell \log^5q)\right)$$ 
where $\ell$ runs through primes dividing $n$ but not $q(q-1)$ with $n_\ell$ the highest power of $\ell$ dividing $n$. Evidently, our run time depends on the prime factorization of $n$. Although at worst quadratic in $n$, we next argue it is subquadratic for \textit{most} $n$. 
% For instance, the natural density of $n$ with exponent $1.81$ is at least $9/10$ and the natural density of $n$ with linear exponent is at least $3/10$. \\ \\ 
%where $n= \prod_\ell \ell^{n_\ell}$ is the factorization of $n$ into prime powers.\\ \\% Throughout, we use $\widetilde{O}()$ to suppress $n^{o(1)}$ and $\log^{o(1)} q$ terms.
%\\ \\
% The maximum is taken over primes $\ell$ dividing $n$ but not $q(q-1)$.\\ \\
If $n$ has a large (say $\Omega(n)$) prime factor not dividing $q(q-1)$, our running time exponent in $n$ is $2$. In all other cases, it is less than $2$. %We improve on the quadratic exponent for \textit{most} $n$ since $n$ for which our exponent is $1.81$ has density at least $9/10$. Further, the density of $n$ for which the exponent is nearly linear is at least $3/10$. These density results are inferred from the distribution of smooth numbers. 
Call $n$ with largest prime factor at most $n^{1/c}$ as $n^{1/c}$-powersmooth. For $n^{1/c}$-powersmooth $n$ with $1<c\leq 2$, our run time exponent in $n$ is at most $2/c$. The natural density of $n^{1/c}$-powersmooth $n$ tends to the Dickman-de Bruin function $\rho(c)$ and for $1 < c \leq 2$, $\rho(c) = 1-\log c$ \cite{gra}. 
In particular, $n^{1/1.1}$-powersmooth $n$ have density $1- \log(1.1) > 9/10$. Hence the $n$ with run time exponent in $n$ at most $2/1.1 \approx 1.8$ have density at least $9/10$. Likewise,  $n^{1/2}$-powersmooth $n$ have density at least $3/10$. Hence the $n$ with run time linear in $n$ have density at least $3/10$.\\ \\
The paper is organized as follows. 
In \S~\ref{preliminary}, the \textsc{Isomorphism Problem} is reduced in linear time to subproblems, each one corresponding to a prime power $\ell^{n_\ell}$ dividing $n$. A key component in the reduction is a fast linear algebraic algorithm (Lemma \ref{toeplitz_lemma}) that takes a polynomial relation between two $\alpha,\beta \in \overline{\F}_q$ of the same degree and computes a root of the minimal polynomial of $\alpha$ in $\F_q(\beta)$. 
In \S~\ref{kummer_section}, subproblems corresponding to prime powers $\ell^{n_\ell}$ such that $\ell$ divides $q-1$ are solved in linear time using Kummer theory. Likewise, in \S~\ref{artin_schreier_section}, subproblems corresponding to powers of the characteristic $p$ are solved in linear time using Artin-Schreier theory. The key in both these special cases is a new recursive algorithm to evaluate the action of idempotents in the Galois group ring that appear in the proof of Hilbert's theorem 90. 
In \S~\ref{elliptic_section}, the generic case of a prime power $\ell^{n_\ell}$ where $\ell \nmid (q-1)p$ is handled using an elliptic curve $E/\F_q$ with $\F_q$ rational $\ell$ torsion. The analogue of Hilbert's theorem 90 in this context is Lang's theorem which states that the first cohomology group $H^1(\F_q,E)$ is trivial \cite{lan}. 
In \S~\ref{elliptic_discrete_log_subsection}, the \textsc{isomorphism problem} is reduced to computing discrete logarithms in the $\F_q$ rational $\ell$ torsion subgroup of $E$. The crux of the reduction is to compute a preimage of a non trivial $\F_q$ rational $\ell$ torsion point under the Lang map. In \S~\ref{lang_subsection}, we devise a fast algorithm to compute such a preimage using $\ell$ isogenies and solve the \textsc{Isomorphism Problem} of degree $\ell^{n_\ell}$ in $\ell^{n_\ell+1+o(1)} \log^{1+o(1)} q + O(\ell \log^5 q)$ time.  
In \S~\ref{lang_subsection}, we pose an algorithmic Problem \ref{problem_lang} concerning Lang's theorem, a solution to which would solve the \textsc{Isomorphism Problem} in subquadratic time for all $n$.\\ \\
% We pose this preimage computation as Problem \ref{problem_lang}, a solution to which would solve \textsc{Isomorphism Problem} in linear time. We hope \\ \\
%In fact our techniques extend to the problem of computing given a polynomial over a finite field $\F_q$ and a positive integer $n$, finds a root in $\F_{q^n}$.
Fast modular composition and fast modular power projection \cite{ku}, key ingredients in our algorithm, are considered impractical with no existing implementations. Practical implications of our algorithm are thus unclear.\\ \\% We sketch a practical quadratic run time variant of our algorithm without fast modular composition.\\ \\
We also extend our algorithm to solve the following more general root finding problem: given a polynomial over $\F_q$ and a positive integer $n$, find its roots in $\F_{q^n}$ (see Remark \ref{root_finding_remark}). The construction of $\F_{q^n}$ could be given or left to the algorithm. The former allows one to compute embeddings of one finite field in another.
%The algorithm of Rains \cite{rai} is currently the fastest in practice. 
\section{Reduction of the \textsc{Isomorphism Problem} to Prime Power Degrees}\label{preliminary}
For $\alpha \in \overline{\F}_q$, call $[\F_q(\alpha):\F_q]$ the degree of $\alpha$. For $\alpha,\beta \in \overline{\F}_q$, call $\alpha \sim \beta$ if and only if there is an integer $j$ such that $\alpha = \sigma^j(\beta)$. That is, $\alpha \sim \beta$ means they have the same minimal polynomial.
\begin{lemma}\label{toeplitz_lemma}
There is an %$\widetilde{O}(n \log q)$
$n^{1+o(1)} \log^{1+o(1)} q$ time algorithm
that given the minimal polynomial $g(x) \in \F_q[x]$ of an $\alpha \in \overline{\F}_q$ of degree $n$ and $f_1(x),f_2(x) \in\F_q[x]$ of degree less than $n$ such that $f_1(\alpha)$ is of degree $n$, finds an $r(x) \in \F_q[x]$ such that $r(\beta)$ is a root of $g(x)$ for all $\beta \in \overline{\F}_q$ satisfying $f_1(\alpha) \sim f_2(\beta)$.
\end{lemma}
\begin{proof}
Since $\alpha$ and $f_1(\alpha)$ both have degree $n$, $\alpha$ is in $\F_q(f_1(\alpha))$ and there is a unique $h(x)=\sum_{i=0}^{n-1}h_ix^i \in \F_q[x]$ such that $h(f_1(\alpha))= \alpha$. We next describe how to compute $h(x)$. \\ \\
Pick $u \in \F_q^n$ uniformly at random and consider the $\F_q$-linear functional $$\mathcal{U}:\F_q(\alpha) \longrightarrow \F_q, y \longmapsto u^ty $$   
%$$\ \ \ \ \ \ \ \ \ \ \ \ y \longmapsto u^ty $$
where $y \in \F_q^n$ is an element of $\F_q(\alpha)$ written in the standard basis $(1,\alpha,\alpha^2,\ldots, \alpha^{n-1})$.\\ \\
Abusing notation, let $h = (h_0,h_1,\ldots,h_{n-1})^t$ denote the coefficient vector of $h(x)$. We will determine $h(x)$ by solving the  linear system $$\mathcal{U}(\alpha^i h(f_1(\alpha))) = \mathcal{U}(\alpha f_1(\alpha)^i), i \in \{0,1,\ldots,2n-2\}$$
in its coefficients. Let $A$ be the $n$ by $2n-1$ matrix whose $i^{th}$ column consists of $f_1(\alpha)^{i-1}$ written in the standard basis. Multiplication by $\alpha$ is an $\F_q$ linear transformation on $\F_q(\alpha)$ with matrix representation on the standard basis being the companion matrix 
\[X:=
\begin{bmatrix}
    0 & 0 & 0 & \dots  & 0 & -g_0 \\
    1 & 0 & 0 & \dots  & 0 & -g_1 \\
    0 & 1 & 0 & \dots  & 0 & -g_2 \\
    \vdots & \vdots & \vdots & \ddots & \vdots & \vdots \\
    0 & 0 & 0 & \dots  & 1 & -g_{n-1} 
\end{bmatrix}
\] 
with respect to $g(x)=\sum_{i=0}^{n-1}g_ix^i + x^n$. Let $(a_0,a_1,\ldots,a_{2n-2}):= u^tA$ and $(b_0,b_1,\ldots,b_{2n-2}):= u^tXA$. Since $X$ has at most $2n-1$ non zero coefficients, $u^tX$ can be computed with number of $\F_q$-operations bounded linearly in $n$. Given $u$, $f_1(x)$ and $g(x)$, to compute $u^tA$ is an instance of the modular power projection problem. Likewise computing $u^tXA$ given $u^tX$, $f_1(x)$ and $g(x)$. By \cite{ku}, each of these modular power projection instances can be solved in $n^{1+o(1)} \log^{1+o(1)} q)$ time. The aforementioned linear system in matrix form is 
%\[
\begin{equation}\label{toeplitz_equation}
\begin{bmatrix}
    a_0 & a_1 & a_2 & \dots  & a_{n-1} \\
    a_1 & a_2 & a_3 & \dots  & a_n  \\
    a_2 & a_3 & a_4 & \dots  & a_{n+1} \\
    \vdots & \vdots & \vdots & \ddots  & \vdots \\
    a_{n-1} & a_{n} & a_{n+1} & \dots  & a_{2n-2} 
\end{bmatrix}
\begin{bmatrix}
h_0\\h_1\\ h_3\\ \vdots \\ h_{n-1}
\end{bmatrix}
=
\begin{bmatrix}
b_0\\b_1\\b_2\\ \vdots\\ b_{2n-2}
\end{bmatrix}
%\]
\end{equation}
and by \cite{sho_mp} has full rank with probability at least $1/2$ for a randomly chosen $u$. One of its solutions is the coefficient vector $h$ of the $h(x)$ we seek. Being Toeplitz, in $n^{1+o(1)}\log^{1+o(1)} q)$ time, we can test if it is full rank and if so find the solution $h$. Once $h(x)$ is found, using \cite[Corollary 1]{cl} to compose polynomials, within time stated in the lemma, we output $h(f_2(x))$ as $r(x)$. The output is correct since $h(f_2(\beta)) \sim h(f_1(\alpha)) = \alpha$.\end{proof}

\begin{lemma}\label{subfield_projection_lemma}
There is an algorithm that given the minimal polynomial $g(x) \in \F_q[x]$ of an $\alpha \in \overline{\F}_q$ of degree $m$ and a positive integer $n$ dividing $m$, finds an element $\alpha_n \in \F_q(\alpha)$ of degree $n$ and its minimal polynomial over $\F_q$ in  time $m^{1+o(1)} \log^{2+o(1)} q$.
\end{lemma}
\begin{proof}
Pick $\beta \in \F_q(\alpha)$ uniformly at random and set $\alpha_n:=\sum_{i=0}^{m/n-1}\sigma^{ni}(\beta)$, the trace of $\beta$ down to $\F_{q^n} \subseteq \F_q(\alpha)$. By iterated Frobenius \cite{gs}\cite{ku}, this trace computation can be performed in the time stated in the lemma. Compute the minimal polynomial $M(x)\in \F_q[x]$ of $\alpha_n$ over $\F_q$ using \cite{sho_mp}\cite[\S~8.4]{ku}, again, in time stated in the lemma. If the degree of $M(x)$ is $n$, output $\alpha_n$ and $M(x)$. Since the trace down to $\F_{q^n}$ maps a random element from $\F_q(\alpha)$ to a random element in $\F_{q^n}$, we succeed with probability at least $1/2$.
\end{proof}
\noindent We next reduce \textsc{Isomorphism Problem} to itself restricted to prime power input degree.
\begin{lemma}\label{prime_power_lemma} 
Let $n=\prod_{\ell} \ell^{n_\ell}$ be the factorization of $n$ into prime powers. In $n^{1+o(1)} \log^{2+o(1)} q$ time, \textsc{Isomorphism Problem} with inputs of degree $n$ may be reduced to identical problems; one for each prime $\ell$ dividing $n$ with inputs of degree $\ell^{n_\ell}$.
\end{lemma}
\begin{proof}
Consider an input $f(x),g(x)\in \F_q[x]$ to \textsc{Isomorphism Problem}. Let $\alpha,\beta \in \overline{\F}_q$ respectively be roots of $f(x),g(x)$. Compute the factorization $n=\prod_{\ell} \ell^{n_\ell}$ of $n$ into prime powers. For each prime $\ell$ dividing $n$, using Lemma \ref{subfield_projection_lemma}, compute $\alpha_\ell \in \F_q(\alpha)$ and $M_\ell(x) \in \F_q[x]$ such that $\alpha_\ell$ has degree $\ell^{n_\ell}$ and $M_\ell$ is the minimal polynomial of $\alpha_\ell$. Likewise compute $\beta_\ell \in \F_q(\beta)$ and $N_\ell(x) \in \F_q[x]$ such that $\beta_\ell$ has degree $\ell^{n_\ell}$ and $N_\ell$ is the minimal polynomial of $\beta_\ell$. 
Since $\F_{q^{\ell^{n_\ell}}}$ and $\F_{q^{n/\ell^{n_\ell}}}$ are linearly disjoint over $\F_q$, both $\sum_{\ell | n} \alpha_\ell$ and $\sum_{\ell | n}\beta_\ell$ have degree $n$. For each $\ell$ dividing $n$, solve \textsc{Isomorphism Problem} with input $M_\ell(x),N_\ell(x)$ and find a root $\beta^\prime_\ell$ of $M_\ell(x)$ in $\F_q(\beta_\ell)$. Now for all $\ell$ dividing $n$, $\alpha_\ell \sim \beta^\prime_\ell$. Applying Lemma \ref{toeplitz_lemma} to the relation $\sum_{\ell | n }\alpha_\ell \sim \sum_{\ell|n}\beta^\prime_\ell$, we solve \textsc{Isomorphism Problem} with input $f(x),g(x)$.\end{proof}

\begin{remark}\label{root_finding_remark}
Consider the problem of finding a root of a degree $m$ polynomial $f(x) \in \F_q[x]$ in $\F_{q^n}$, where $\F_{q^n}$ is constructed as $\F_q[x]/(g(x))$ for a monic irreducible $g(x)$. Either $g(x)$ is given or constructed in linear time using \cite{cl}. We show that this problem reduces to the \textsc{Isomorphism Problem} in time linear in $m$ and $n$. In fact, the reduction finds not just one but all the roots of $f(x)$ in an implicit form. The output is a set of roots of $f(x)$ whose orbit under $\sigma$ is the set of all roots of $f(x)$. For $f(x)$ to have a root in $\F_{q^n}$, $f(x)$ has to have an irreducible factor of degree dividing $n$. Since the number of factors of $n$ is at most $\log n$, using \cite{ku}, in $m^{1+o(1)} \log^{2+o(1)} q \log^{1+o(1)} n$ time, we may enumerate all irreducible factors of $f(x)$ of degree dividing $n$. For each such irreducible factor $h(x)$, using Lemma \ref{subfield_projection_lemma}, identify a subfield of $\F_{q^n}$ and find a root $h(x)$ in the subfield by solving the \textsc{Isomorphism Problem}.
%We may choose one among these factors and find its root. Hence without loss of generality, we may assume that $h(x)$ is monic irreducible and that $n$ divides $m$. Further, by Lemma \ref{subfield_projection_lemma}, we may identify a subfield $\F_{q^m} \subseteq \F_{q^n}$ to compute roots in. In summary, the problem of finding a root of a degree $m$ polynomial $f(x) \in \F_q[x]$ in $\F_{q^n}$ reduces to \textsc{Isomorphism Problem}. \\ \\
%For a positive integer $d$, let $\F_{q^d}$ denote the degree $d$ extension of $\F_q$ in $\overline{\F}_q$. The isomorphism problem \ref{root_finding_problem} is a special case of finding a root of polynomial $f(x) \in \F_q[x]$ in $\F_{q^n}$. We obtain a fast algorithm for this more general root finding problem since it reduces to the isomorphism problem \ref{root_finding_problem} in nearly linear time (see remark \ref{root_finding_remark}). \\ \\
%In fact, we find not just one but all roots (implicitly). The output is a subset of $\F_{q^n}$ whose orbit under $\sigma$ is the set of all roots of $f(x)$.% The algorithm assumes that $\F_{q^n}$ is constructed as $\F_q(\alpha)$ where $\alpha \in \overline{\F}_q$ is a root of a given monic irreducible $g(x) \in \F_q[x]$. This assumption is without loss of generality for if $g(x)$ is not given but instead its degree $n$ is, then a degree $n$ irreducible polynomial can be constructed in time nearly linear in $n$ \cite{cl}.
% If $g(x)$ is not given but instead its degree $n$ is, then a degree $n$ irreducible $g(x)$ can be constructed in time nearly linear in $n$ \cite{cl}.
\end{remark}

\section{Root Finding in Kummer Extensions of Finite Fields}\label{kummer_section}
Using Kummer theory, we solve the \textsc{Isomorphism Problem} restricted to the case when $n$ is a power of a prime $\ell$ dividing $q-1$. The novelty here is a fast recursive evaluation of the idempotent appearing in the standard proof of (cyclic) Hilbert's theorem 90. 
\begin{lemma}\label{hilbert_lemma} There is an algorithm that given a finite extension $L/\F_q$, an integer $m \leq [L:\F_q]$ and a $\zeta \in L$ such that $\zeta \in K:=\{\beta \in L|\sigma^m(\beta)=\beta\}$ and $\zeta^{[L:K]}=1$, finds an $\alpha \in L$ such that $\sigma^m(\alpha)=\zeta\alpha$ in $[L:\F_q]^{1+o(1)} \log^{2+o(1)} q$ time.
\end{lemma}
\begin{proof}
Since the norm of $\zeta$ from $L$ down to $K$ is $\zeta^{[L:K]}=1$, an $\alpha$ as claimed in the lemma exists by Hilbert's theorem $90$ applied to the cyclic extension $L/K$. We next describe an algorithm that finds such an $\alpha$ in the stated time.\\ \\
Define $\tau:=\zeta^{-1}\sigma^m$, viewed as a $K$-linear endomorphism on $L$. By independence of characters, $\sum_{i=0}^{[L:K]-1}\tau^i$ is non zero. Pick $\theta \in L$ uniformly at random. If $\sum_{i=0}^{[L:K]-1} \tau^i(\theta) \neq 0$ (which happens with probability at least $1/2$), set $\alpha = \sum_{i=0}^{[L:K]-1} \tau^i(\theta)$. Since $\zeta^{-1} \in \F_q$ and $\zeta^{-[L:K]}=1$, $$ \tau(\alpha)= \sum_{i=0}^{[L:K]-1}\zeta^{-i}\sigma^{mi}(\alpha) = \alpha \Rightarrow \tau(\alpha) = \alpha \Rightarrow \zeta^{-1}\sigma^m(\alpha) = \alpha \Rightarrow \sigma^m(\alpha)=\zeta \alpha. $$
We next demonstrate $\sum_{i=0}^{[L:K]-1}\tau^i(\theta)$ can be computed fast given $\theta \in L$. Our approach is similar to the iterated Frobenius trace computation of von zur Gathen and Shoup \cite{gs}.\\ \\% using fast modular composition \cite{ku}.
Let $L$ be given as $\F_q(\eta)$ for some $\eta \in \overline{\F}_q$ with minimal polynomial $g(x) \in \F_q[x]$. By repeated squaring, in time $\widetilde{O}([L:\F_q]\log^2 q)$ compute $\eta^q$. 
For a positive integer $b$, let $\Sigma_b$ denote the partial sum $ \sum_{i=0}^{b-1}\tau^{i}(\theta)$. Our goal is to compute $\Sigma_{[L:K]}$. For every positive integer $b$, 
$$\sum_{i=0}^{2b-1}\tau^{i}(\theta) = \sum_{i=0}^{b-1}\tau^{i}(\theta) + \sum_{i=b}^{2b-1}\tau^{i}(\theta) = \sum_{i=0}^{b-1}\tau^{i}(\theta) + \tau^{b}\left( \sum_{i=0}^{2b-1}\tau^{i}(\theta)\right)$$
\begin{equation}\label{iterated_hilbert}
\Rightarrow \sum_{i=0}^{2b-1}\tau^{i}(\theta) = \sum_{i=0}^{b-1}\tau^{i}(\theta) + \zeta^{-b}\sigma^{b}\left( \sum_{i=0}^{b-1}\tau^i(\theta)\right) \Rightarrow \Sigma_{2b} = \Sigma_b+\zeta^{-b} \sigma^{bm}(\Sigma_b).
\end{equation}
Given $\Sigma_b$ and $\eta^q$, $\sigma^{bm}(\Sigma_b)$ can be computed in $[L:\F_q]^{1+o(1)} \log^{1+o(1)} q$ time using the Frobenius representation of \cite{gs} and fast modular composition \cite{ku}. Hence, given $\Sigma_b$, computing $\Sigma_{2b}$ using equation \ref{iterated_hilbert} takes $[L:\F_q]^{1+o(1)} \log^{1+o(1)} q$ time, which evidently is independent of $b$ and $m$.\\ \\ %This running time is independent of $b$ and $m$.
Set $c = \lfloor \log_2[L:K] \rfloor$ and compute $\Sigma_{2^c}$ by successively computing $\Sigma_0,\Sigma_2,\Sigma_4,\ldots, \Sigma_c$ using equation \ref{iterated_hilbert}. Since $c \leq \log_2[L:K]$, this takes $[L:\F_q]^{1+o(1)} \log^{1+o(1)} q$ time. If $[L:K]$ is not a power of $2$, we recursively compute $\Sigma_{[L:K]-c}$. With the knowledge of $\Sigma_c$ and $\Sigma_{[L:K]-c}$, $\Sigma_m$ may be computed in $\widetilde{O}([L:\F_q] \log q)$ time \cite{gs,ku} as 
\begin{equation}
\Sigma_{[L:K]} = \Sigma_c + \zeta^{-c}\sigma^{mc}(\Sigma_{[L:K]-c}).
\end{equation}
Since $[L:K]-c \leq [L:K]/2$, at most $\log_2[L:K]$ recursive calls are made in total.
\end{proof}
We next state the algorithm followed by proof of correctness and implementation details.
\begin{algorithm}[H]
\caption{Root Finding Through Kummer Theory:}\label{kummer_algorithm}
\begin{algorithmic}[1]
\INPUT Monic irreducibles $g_1(x),g_2(x) \in \F_q[X]$ of degree $\ell^a$ where $\ell$ is a prime dividing $q-1$ and $a$ is a positive integer.
%\OUTPUT A root of $g_1(x)$ in $\F_q[x]/(g_2(x))$ and a root of $g_2(x)$ in $\F_q[x]/(g_1(x))$.
\OUTPUT A root of $g_1(x)$ in $\F_q(\beta_2)$ where $\beta_2 \in \overline{\F}_q$ is a root of $g_2(x)$. 
\State Find a primitive $\ell^{th}$ root of unity $\zeta_\ell \in \F_q$.
\State Construct $\F_q(\beta_1) \cong \F_{q^{\ell^a}}$ where $\beta_1$ is a root of $g_1(x)$. 
\LineComment Apply lemma \ref{hilbert_lemma} with $\left(L=\F_q(\beta_1),m=\ell^{a-1},\zeta=\zeta_\ell\right)$ and find $\alpha_1 \in \F_q(\beta_1)$ such that $$\sigma^{\ell^{a-1}}(\alpha_1) = \zeta_\ell\alpha_1.$$
\LineComment Compute $\alpha_1^\ell$.  \textit{($\alpha_1^\ell$ will have degree $\ell^{a-1}$.)}
\State Construct $\F_q(\beta_2) \cong \F_{q^{\ell^a}}$ where $\beta_2$ is a root of $g_2(x)$. 
\LineComment Apply lemma \ref{hilbert_lemma} with $\left(L=\F_q(\beta_2),m=\ell^{a-1},\zeta=\zeta_\ell\right)$ and find $\alpha_2 \in \F_q(\beta_2)$ such that $$\sigma^{\ell^{a-1}}(\alpha_2) = \zeta_\ell\alpha_2.$$
\LineComment Compute $\alpha_2^\ell$. \textit{($\alpha_2^\ell$ will have degree $\ell^{a-1}$.)}
\State If $a=1$,
\LineComment Find an $e \in \F_q$ such that $e^\ell=\alpha_1^\ell/\alpha_2^\ell$. 
\LineComment Apply lemma \ref{toeplitz_lemma} to $\alpha_1 \sim e \alpha_2$ and find a root of $g_1(x)$ in $\F_q(\beta_2)$. %and a root of $g_2(x)$ in $\F_q(\alpha_1)$.
%\LineComment Apply lemma \ref{toeplitz_lemma} to $\alpha_1 = e \alpha_2$ and find a root of $g_1(x)$ in $\F_q(\alpha_2)$ and a root of $g_2(x)$ in $\F_q(\alpha_1)$.
\State If $a\neq 1$,
\LineComment Find the minimal polynomials $h_1(x),h_2(x)$ over $\F_q$ of $\alpha_1^\ell,\alpha_2^\ell$  respectively. 
\LineComment Recursively find a root $\alpha$ of $h_1(x)$ in $\F_q(\alpha_2^\ell) = \F_q[x]/(h_2(x))$. \textit{($h_1(x)$ and $h_2(x)$ have degree $\ell^{a-1}$.)}
\LineComment Find a $\gamma \in \F_q(\alpha_2^\ell)$ such that $\gamma^\ell = \alpha/\alpha_2^\ell$.
\LineComment Apply lemma \ref{toeplitz_lemma} to $\alpha_1 \sim \gamma \alpha_2$ and find a root of $g_1(x)$ in $\F_q(\beta_2)$.% and a root of $g_2(x)$ in $\F_q(\alpha_1)$.
%\LineComment Apply lemma \ref{toeplitz_lemma} to $\alpha_1 \sim \gamma \alpha_2$ and find a root of $g_1(x)$ in $\F_q(\alpha_2)$ and a root of $g_2(x)$ in $\F_q(\alpha_1)$.
\end{algorithmic}
\end{algorithm}
We next argue that Algorithm \ref{kummer_algorithm} runs to completion and is correct. \\ \\
Since $\ell$ divides $q-1$, there is a primitive $\ell^{th}$ root of unity in $\F_q$, as required in Step $1$.\\ \\
In Step $2$, $\alpha_1^\ell$ is claimed to have degree $\ell^{a-1}$. Let $b$ be the degree of $\alpha_1^\ell$. Since $$\sigma^{\ell^{a-1}}(\alpha_1^\ell)= \left(\sigma^{\ell^{a-1}}(\alpha_1)\right)^\ell = \zeta_\ell^\ell \alpha_1^\ell = \alpha_1^\ell,$$ 
$b$ divides $\ell^{a-1}$. Since $\sigma^{b}(\alpha_1^\ell) = \alpha_1^\ell$, $\sigma^{b}(\alpha_1)/\alpha_1$ is an $\ell^{th}$ root of unity. Thus $\sigma^{\ell b}(\alpha_1)=\alpha_1$ implying the degree of $\alpha_1$ divides $\ell b$. Since $\zeta_\ell \neq 1$, $\alpha_1$ has degree $\ell^a$. Thus $\ell^{a}$ divides $b\ell$ and we may conclude that $\alpha_1^\ell$ has degree $\ell^{a-1}$. Likewise, in Step $3$, $\alpha_2$ has degree $\ell^{a-1}$.\\ \\
In Step $4$, since $a=1$, $\alpha_1^\ell,\alpha_2^\ell \in \F_q$. Further $\alpha_1/\alpha_2 \in \F_q$ since $\sigma(\alpha_1/\alpha_2) = (\zeta \alpha_1)/(\zeta\alpha_2) = \alpha_1/\alpha_2$. Thus
$\alpha_1^\ell/\alpha_2^\ell$ is an $\ell^{th}$ power in $\F_q$ ensuring that an $e \in \F_q$ such that $e^\ell = \alpha_1^\ell/\alpha_2^\ell$ exists. Hence $(\alpha_1/a\alpha_2)$ is an $\ell^{th}$ root of unity and there exists an integer $i$ such that $\alpha_1 = \sigma^{i(\ell^{a-1})}(e \alpha_2)$. Further $\alpha_1$ has degree $\ell^a$. Hence Lemma \ref{toeplitz_lemma}, when applied to the relation $\alpha_1 \sim e \alpha_2$, correctly finds the desired output.\\ \\
The recursive call in Step $5$ yields a root $\alpha\in \F_q(\alpha_2^\ell)$ of $h_1(x)$. Hence $\alpha = \sigma^{j}(\alpha_1^\ell) = (\sigma^j(\alpha_1))^\ell $ for some integer $j$. Further, $\sigma^j(\alpha_1)/\alpha_2 \in \F_q(\alpha_2^\ell)$ since
 $$ \sigma^{\ell^{a-1}}(\sigma^j(\alpha_1)/\alpha_2)=\sigma^j(\zeta_\ell\alpha_1)/(\zeta_\ell\alpha_2) = \sigma^j(\alpha_1)/\alpha_2.$$
Hence $\alpha/\alpha_2^\ell = (\sigma^j(\alpha_1)/\alpha_2)^\ell$ is an $\ell^{th}$ power in $\F_q(\alpha_2^\ell)$ assuring the existence of a $\gamma$ that is sought in Step $5$. For such a $\gamma$, $\gamma^\ell = (\sigma^j(\alpha_1)/\alpha_2)^\ell$ implying $\sigma^j(\alpha_1)/(\gamma \alpha_2)$ is an $\ell^{th}$ root of unity. Hence, there exists an integer $i$ such that $$\sigma^j(\alpha_1) = \gamma \sigma^{i\ell^{a-1}}(\alpha_2) = \sigma^{i\ell^{a-1}}(\gamma\alpha_2).$$ Further $\alpha_1$ has degree $\ell^a$. Hence Lemma \ref{toeplitz_lemma}, when applied to the relation $\alpha_1 \sim \gamma c \alpha_2$, correctly finds the desired output.
\subsection{Implementation and Running Time Analysis}
To implement Step $1$, pick a random $c\in \F_q$ and if $c^{\frac{(q-1)}{\ell}} \neq 1$, set $\zeta = c^{\frac{q-1}{m}}$. Else try again with a new independent choice $c \in \F_q$. We succeed in finding a $\zeta$ if the $c$ chosen is not a $\ell^{th}$ power. This happens with probability at least $1-1/\ell$. The expected running time of Step $1$ is hence $O(\log^2 q)$. Running times of Steps $2$ and $3$ are dominated by the $\ell^{a+o(1)}\log^{2+o(1)} q$ time their respective calls to Lemma \ref{hilbert_lemma} take.\\ \\
In Step $4$, find a root $a$ of $x^\ell-(\alpha_1^\ell/\alpha_2^\ell) \in \F_q[x]$ in $\F_q$ using \cite{gs,ku} in $\ell^{1+o(1)}\log^{2+o(1)} q$ time. The invocations to Lemma \ref{toeplitz_lemma} in Steps $4$ and $5$ each take $\ell^{a+o(1)}\log^{2+o(1)} q$ time.\\ \\
In Step $5$, minimal polynomials of $\alpha_1$ and $\alpha_2$ can be computed in $\ell^{a+o(1)}\log^{1+o(1)} q)$ time \cite[\S~8.4]{ku}. To compute $\gamma$, find a root of $x^\ell-\alpha/\alpha_2^\ell$ in $\F_q(\alpha_2^\ell) = \F_q[x]/(h_2(x))$ using \cite{gs,ku}. Since we a finding the root of a degree $\ell$ polynomial over a field of size $q^{\ell^{a-1}}$, the running time $\ell^{a+o(1)}\log^{1+o(1)} q)$ turns out to be nearly linear in $\ell^{a}$.\\ \\
Algorithm \ref{kummer_algorithm} makes at most one recursive call to an identical subproblem of size $\ell^{a-1}$. Hence at most $a$ recursive calls are made in total. In summary, we have the following theorem. 
\begin{theorem}\label{kummer_theorem}
Algorithm \ref{kummer_algorithm} solves the \textsc{Isomorphism Problem} restricted to the special case when $n$ is a power of a prime $\ell$ dividing $q-1$ in $n^{1+o(1)} \log^{2+o(1)} q$ time.
\end{theorem}
%\begin{theorem}\label{kummer_theorem}
%Algorithm \ref{kummer_algorithm}, given monic irreducibles $g_1(x),g_2(x) \in \F_q[X]$ of degree $\ell^a$ where $\ell$ is a prime dividing $q-1$, finds a root of $g_1(x)$ in $\F_q[x]/(g_2(x))$ and a root of $g_2(x)$ in $\F_q[x]/(g_1(x))$ in $(\ell^a)^{1+o(1)} \log q + \ell^a (\log q)^{2+o(1)}$ expected time.
%\end{theorem}
%\begin{remark} Algorithm \ref{kummer_algorithm} To Do. can be implemented in $m^{2+o(1)}(\log q)^{2+o(1)}$ time without the fast modular composition algorithm \cite{ku}. The only places where the fast modular composition was used are Lemma\ref{toeplitz_lemma} , Lemma \ref{hilbert_lemma} and the root finding in Steps $4$ and $5$. Without fast modular composition, the root finding in Steps $4$ and $5$ can be computed in $(\ell^{a})^{1+o(1)} (\log q)^{1+o(1)} + \ell^a  (\log q)^{2+o(1)}$ time \cite{gs}. By remark ?, the algorithm in Lemma\ref{toeplitz_lemma} can be implemented in $(\ell^{a})^{2+o(1)}\log q$ time without fast modular composition. In the proof of Lemma \ref{hilbert_lemma}, we may compute $\sum_{i=0}^{m-1}\zeta^{-i}\sigma^i(\theta)$ as follows. Iteratively compute the $\{\theta,\theta^q,\ldots,\theta^{q^{m-1}}\}$. That is, given $\theta^{q^i}$, $\theta^{q^{i+1}}$ is computed as $(\theta^{q^i})^q$ using repeated squaring HOW ?. From this, compute $\sum_{i=0}^{m-1}\zeta^{-i}\sigma^i(\theta)$ with an overall running time $({\ell^a})^{2+o(1)} (\log q)^{2+o(1)}$.\\ \\
%\end{remark}
\section{Root Finding in Artin-Schreier Extensions of Finite Fields}\label{artin_schreier_section}
Using Artin-Schrier theory, we solve the \textsc{Isomorphism Problem} restricted to the special case when $n$ is a power of the characteristic $p$. The novelty here is a fast recursive evaluation of the idempotent in the proof of the additive version of (cyclic) Hilbert's theorem 90. 
\begin{lemma}\label{hilbert_additive_lemma} There is an algorithm that given a finite extension $L/\F_q$ of degree $[L:\F_q]$ divisible by $p$, finds an $\alpha \in L$ such that $\sigma^{[L:\F_q]/p}(\alpha)=\alpha+1$ in $[L:\F_q]^{1+o(1)} \log^{2+o(1)} q$  time. 
\end{lemma}
\begin{proof}
Let $m:=[L:K]/p$ and $K:=\{\beta \in L|\sigma^m(\beta)=\beta\}$.
Since the trace of $1$ from $L$ down to $K$ is $0$, an $\alpha$ as claimed in the lemma exists by Hilbert's theorem $90$ applied to the cyclic extension $L/K$. We next describe an algorithm that finds such an $\alpha$ in the stated time. \\ \\
Let $Tr_{L/K} = \sum_{i=0}^{p-1}\sigma^{mi}$ denote the trace from $L$ to $K$. Pick $\theta \in L$ uniformly at random. If $Tr_{L/K}(\theta) \neq 0$  (which happens with probability at least $1/2$), setting $$\alpha := \frac{-1}{Tr_{L/K}(\theta)}\sum_{i=0}^{p-1} i\sigma^{mi}(\theta)$$ ensures $\sigma^{m}(\alpha)-\alpha = 1$.
We next demonstrate that given $\theta \in L$, $\alpha$ can be computed fast. \\ \\
Let $L$ be given as $\F_q(\eta)$ for some $\eta \in \overline{\F}_q$ with minimal polynomial $g(x) \in \F_q[x]$. By repeated squaring, in time $O([L:\F_q]\log^2 q)$ compute $\eta^q$.\\ \\
For a positive integer $b$, let $\Sigma_b$ denote the partial sum $ \sum_{i=0}^{b-1}i\sigma^{mi}(\theta)$ and let $\Gamma_b$ denote the partial trace $\sum_{i=0}^{b-1}\sigma^{mi}(\theta)$. We intend to compute $\Sigma_{p}$ and $\Gamma_p$ to set $\alpha = \Sigma_p/\Gamma_p$.\\ \\
For every positive integer $b$, 
$$\sum_{i=0}^{2b-1}i\sigma^{mi}(\theta) = \sum_{i=0}^{b-1}i\sigma^{mi}(\theta) + \sum_{i=b}^{2b-1}i\sigma^{mi}(\theta) = \sum_{i=0}^{b-1}i\sigma^{mi}(\theta) +  \sum_{i=0}^{b-1}(b+i)\sigma^{m(b+i)}(\theta)$$
$$= \sum_{i=0}^{b-1}i\sigma^{mi}(\theta) + b\sigma^{mb}\left( \sum_{i=0}^{b-1}\sigma^{mi}(\theta)\right) + \sigma^{mb}\left( \sum_{i=0}^{b-1}i\sigma^{mi}(\theta)\right).$$
\begin{equation}\label{iterated_hilbert_additive}
\Rightarrow \Sigma_{2b} = \Sigma_b + b \sigma^{bm}(\Sigma_b) + \sigma^{bm}(\Gamma_b).
\end{equation}
Likewise
\begin{equation}\label{iterated_trace_additive}
\Gamma_{2b} = \Gamma_b + \sigma^{bm}\Gamma_b.
\end{equation}
Given $\Sigma_b, \Gamma_b$ and $\eta^q$, $\sigma^{mb}(\Sigma_b)$ and $\sigma^{mb}(\Gamma_b)$ can be computed in $[L:\F_q]^{1+o(1)} \log^{2+o(1)} q)$ time using the Frobenius representation of \cite{gs} and fast modular composition \cite{ku}. Hence, given $\Sigma_b$ and $\Gamma_b$, computing $\Sigma_{2b}$ and $\Gamma_{2b}$ using equations \ref{iterated_hilbert_additive} and \ref{iterated_trace_additive} takes $[L:\F_q]^{1+o(1)}\log^{2+o(1)} q$ time. This running time is independent of $b$ and $m$.\\ \\
Set $c = \lfloor \log_2p \rfloor$ and successively compute $\Sigma_0,\Gamma_0, \Sigma_2,\Gamma_2,\Sigma_4,\Gamma_4, \ldots, \Sigma_{2^c},\Gamma_{2^c}$ using equations \ref{iterated_hilbert_additive} and \ref{iterated_trace_additive}. Since $c \leq \log_2p$, this takes $\widetilde{O}([L:\F_q] \log^2 q)$ time. If $p$ is not a power of $2$, we recursively compute $\Sigma_{p-2^c}$ and $\Gamma_{p-2^c}$. With the knowledge of $\Sigma_{2^c},\Gamma_{2^c},\Sigma_{p-2^c},\Gamma_{p-2^c}$, we may compute $\Sigma_p$ and $\Gamma_p$ in $\widetilde{O}([L:\F_q]\log q)$ time as 
\begin{equation}
\Sigma_{p} = \Sigma_{2^c} + 2^c\sigma^{m2^c}(\Sigma_{p-2^c})+\sigma^{m2^c}(\Gamma_{p-2^c}), \Gamma_p = \Gamma_2^c+\sigma^{m2^c}(\Gamma_{p-2^c}).
\end{equation}
Since $p-2^c \leq p/2$, at most $\log_2p$ recursive calls are made in total.
\end{proof}
We next state the algorithm followed by proof of correctness and implementation details.
\begin{algorithm}[H]
\caption{Root Finding Through Artin-Schreier Theory:}\label{artin_schreier_algorithm}
\begin{algorithmic}[1]
\INPUT Monic irreducibles $g_1(x),g_2(x) \in \F_q[X]$ of degree $p^a$ where $a$ is a positive integer.
%\OUTPUT A root of $g_1(x)$ in $\F_q[x]/(g_2(x))$ and a root of $g_2(x)$ in $\F_q[x]/(g_1(x))$.
\OUTPUT A root of $g_1(x)$ in $\F_q(\beta_2)$ where $\beta_2 \in \overline{\F}_q$ is a root of $g_2(x)$. 
\State Construct $\F_q(\beta_1) \cong \F_{q^{p^a}}$ where $\beta_1$ is a root of $g_1(x)$. 
\LineComment Apply Lemma \ref{hilbert_additive_lemma} with $L=\F_q(\beta_1)$ and find $\alpha_1 \in \F_q(\beta_1)$ such that $$\sigma^{p^{a-1}}(\alpha_1) = \alpha_1 + 1.$$
\LineComment Compute $\alpha_1^p-\alpha_1$.  \textit{($\alpha_1^p-\alpha_1$ will have degree $p^{a-1}$.)}
\State Construct $\F_q(\beta_2) \cong \F_{q^{p^a}}$ where $\beta_2$ is a root of $g_2(x)$. 
\LineComment Apply Lemma \ref{hilbert_additive_lemma} with $L=\F_q(\beta_2)$ and find $\alpha_2 \in \F_q(\beta_2)$ such that $$\sigma^{p^{a-1}}(\alpha_2) = \alpha_2 + 1.$$
\LineComment Compute $\alpha_2^p-\alpha_2$.  \textit{($\alpha_2^p-\alpha_2$ will have degree $p^{a-1}$.)}
\State If $a=1$,
\LineComment Find an $e \in \F_q$ such that $e^p-e=(\alpha_1^p-\alpha_1)- (\alpha_2^p-\alpha_2)$. 
%\LineComment Apply Lemma\ref{toeplitz_lemma} to $\alpha_1 = \alpha_2 + e$ and find a root of $g_1(x)$ in $\F_q(\alpha_2)$ and a root of $g_2(x)$ in $\F_q(\alpha_1)$.
\LineComment Apply Lemma \ref{toeplitz_lemma} to $\alpha_1 \sim  \alpha_2+ e$ and find a root of $g_1(x)$ in $\F_q(\beta_2)$.
\State If $a\neq 1$,
\LineComment Find the minimal polynomials $h_1(x),h_2(x)$ over $\F_q$ of $\alpha_1^p-\alpha_1,\alpha_2^p-\alpha_2$  respectively. 
\LineComment Recursively find a root $\alpha$ of $h_1(x)$ in $\F_q(\alpha_2^p-\alpha_2) = \F_q[x]/(h_2(x))$. \textit{($h_1(x)$ and $h_2(x)$ have degree $\ell^{a-1}$.)}
\LineComment Find a $\gamma \in \F_q(\alpha_2^\ell)$ such that $\gamma^p-\gamma = \alpha- (\alpha_2^p-\alpha_2)$.
\LineComment Apply Lemma \ref{toeplitz_lemma} to $\alpha_1 \sim  \alpha_2  + \gamma$ and find a root of $g_1(x)$ in $\F_q(\beta_2)$.
\end{algorithmic}
\end{algorithm}
We next argue that Algorithm \ref{artin_schreier_algorithm} runs to completion and is correct. \\ \\
In Step $1$, $\alpha_1^p-\alpha_1$ is claimed to have degree $p^{a-1}$. Let $b$ be the degree of $\alpha_1^p-\alpha_1$. Since $$\sigma^{p^{a-1}}(\alpha_1^p-\alpha_1)= \left(\sigma^{p^{a-1}}(\alpha_1)\right)^p-\sigma^{p^{a-1}}(\alpha_1) = \alpha_1^p+1-(\alpha_1+1) = \alpha_1^p-\alpha_1,$$ 
$b$ divides $p^{a-1}$. Since $\alpha_1$ is a root of $x^p-x-(\alpha_1^p-\alpha_1)$ and $\alpha_1$ has degree $p^{a}$, $\alpha_1^p-\alpha_1$ has degree at most $p^{a-1}$. Thus $\alpha_1^p-\alpha_1$ has degree $p^{a-1}$. Likewise, in Step $2$, $\alpha_2^p-\alpha_2$ has degree $p^{a-1}$.\\ \\
In Step $3$, since $a=1$, $\alpha_1^p-\alpha_1,\alpha_2^p-\alpha_2 \in \F_q$. Further $\alpha_1-\alpha_2$ is in $\F_q$ since $\sigma(\alpha_1-\alpha_2) = (\alpha_1+1)-(\alpha_2+1) = \alpha_1-\alpha_2$. Thus
$\alpha_1-\alpha_2 \in \F_q$ is a root of $x^p-x-((\alpha_1^p-\alpha_1) - (\alpha_2^p-\alpha_2))$ ensuring that an $e \in \F_q$ such that $e^p-e = (\alpha_1^p-\alpha_1) - (\alpha_2^p-\alpha_2)$ exists. The roots of $x^p-x-((\alpha_1^p-\alpha_1) - (\alpha_2^p-\alpha_2))$ are $\{e,e+1,e+2,\ldots,e+(p-1)\}$. Hence $\alpha_1-\alpha_2=e$. Further $\alpha_1$ has degree $\ell^a$. Hence Lemma \ref{toeplitz_lemma}, when applied to the relation $\alpha_1 \sim  \alpha_2 + e$, correctly finds the desired output.\\ \\
The recursive call in Step $4$ yields a root $\alpha \in \F_q(\alpha_2^p-\alpha_2)$ of $h_1(x)$. Hence $\alpha = \sigma^{j}(\alpha_1^p-\alpha_1) = (\sigma^j(\alpha_1))^p-\sigma^j(\alpha_1) $ for some integer $j$. Further, $\sigma^j(\alpha_1)-\alpha_2 \in \F_q(\alpha_2^p-\alpha_2)$ since
 $$ \sigma^{p^{a-1}}(\sigma^j(\alpha_1)-\alpha_2)=\sigma^j(\alpha_1+1)-(\alpha_2+1) = \sigma^j(\alpha_1)-\alpha_2.$$
Hence $\sigma^j(\alpha_1)-\alpha_2 \in \F_q(\alpha_2^p-\alpha_2)$ is a root of $x^p-x = \alpha-(\alpha_2^p-\alpha_2)$ assuring the existence of $\gamma$ sought in Step $5$. For such a $\gamma$, the roots of $x^p-x-(\alpha - (\alpha_2^p-\alpha_2))$ are $\{\gamma,\gamma+1,\gamma+2,\ldots,\gamma+(p-1)\}$. Hence $\sigma^j(\alpha_1)-\alpha_2=\gamma$. Further $\alpha_1$ has degree $\ell^a$. Hence Lemma \ref{toeplitz_lemma}, when applied to the relation $\alpha_1 \sim  \alpha_2 + \gamma$, correctly finds the desired output.
\subsection{Implementation and Run Time Analysis}
Running times of Steps $1$ and $2$ are dominated by their respective calls to Lemma \ref{hilbert_lemma}, each taking $p^{a+o(1)} \log^{2+o(1)} q$ time.\\ \\ 
In Step $3$, find a root $e$ of $x^p-x-((\alpha_1^p-\alpha_1)-(\alpha_2^p-\alpha_2)) \in \F_q[x]$ in $\F_q$ using \cite{gs,ku} in $p^{1+o(1)}\log^{2+o(1)} q$ time. Invocations to Lemma \ref{toeplitz_lemma} in Steps $3$ and $4$ take $p^{a+o(1)} \log^{2+o(1)} q$ time.\\ \\
In Step $4$, minimal polynomials of $\alpha_1^p-\alpha_1$ and $\alpha_2^p-\alpha_2$ can be computed in $p^{a+o(1)}\log^{1+o(1)} q$ time \cite[\S~8.4]{ku}. To compute $\gamma$, find a root of $x^p-x-(\alpha-(\alpha_2^p-\alpha_2))$ in $\F_q(\alpha_2^p-\alpha_2) = \F_q[x]/(h_2(x))$ using \cite{gs,ku}. Since we a finding the root of a degree $p$ polynomial over a field of size $q^{p^{a-1}}$, the running time $p^{a+o(1)} \log^{1+o(1)} q$ turns out to be nearly linear in $p^{a}$.\\ \\
Algorithm \ref{artin_schreier_algorithm} makes at most one recursive call to an identical subproblem of size $p^{a-1}$. Hence at most $a$ recursive calls are made in total. In summary, we have the following theorem.
\begin{theorem}\label{artin_schreier_theorem}
Algorithm \ref{artin_schreier_algorithm} solves the \textsc{Isomorphism Problem} restricted to the special case when $n=p^a$ in $n^{1+o(1)}\log^{1+o(1)} q$ time.
\end{theorem}
\section{Root Finding over Extensions of Finite Fields using Elliptic Curves}\label{elliptic_section}
We solve the \textsc{Isomorphism Problem} restricted to the case when $n$ is a power $\ell^a$ of a prime $\ell\nmid q(q-1)$ in $\ell^{a+1+o(1)} \log^{1+o(1)} q + O(\ell \log^5 q)$ time. 
Through this section, fix a prime $\ell$ such that $\ell \nmid q(q-1)$, $\sqrt{q} \geq 5\ell^3$ and a positive integer $a$. 
% Given two degree $\ell^a$ monic irreducible polynomials $g_1(x),g_2(x) \in \F_q[x]$, we will demonstrate how to find a root of $g_2(x)$ in $\F_q[x]/(g_1(x))$ in expected time nearly quadratic in $\ell^a$.
\subsection{Elliptic Curves with $\F_q$-rational $\ell$-torsion}
Let $E$ be an elliptic curve over $\F_q$ such that $\ell$ divides $|E(\F_q)|$ but $\ell^2$ does not. Let $\sigma_E:E\longrightarrow E$ denote the $q^{th}$ power Frobenius endomorphism and $t \in \Z$ the trace of $\sigma_E$. The characteristic polynomial $P_E(X):=X^2-tX+q \in \Z[X]$ of $\sigma_E$ factors modulo $\ell$ as $$X^2-tX+q = (X-1)(X-q) \mod \ell.$$
To see why $1$ is a root of $P_E(X)$ modulo $\ell$, observe $P_E(1)=|E(\F_q)|$ and $\ell \mid |E(\F_q)|$. The other root is $q$, since the product of the roots is $q$. By Hensel's lemma, there exists $\lambda,\mu \in \{0,1,\ldots,\ell^{a+1}-1\}$ such that $$X^2-tX+q = (X-\lambda)(X-\mu) \mod \ell^{a+1},$$ where $\lambda=1 \mod \ell$ and $\mu = q \mod \ell$. Hence there exists $P_\lambda,P_\mu \in E[\ell^{a+1}]$, each of order $\ell^{a+1}$ such that $$E[\ell^{a+1}] = \langle P_\lambda\rangle \oplus \langle P_\mu\rangle, \sigma_E(P_\lambda) = \lambda P_\lambda\ and\ \sigma_E(P_\mu)=\mu P_\mu.$$
Since $\lambda = 1 \mod \ell$ and $\ell^{2} \nmid |E(\F_q)|$, $\lambda = 1 + \gamma \ell$ where $\gamma:=(\lambda-1)/\ell \in \Z_{\geq0}$ and $\gcd(\gamma,\ell)=1$. 
\subsection{Root Finding Through Discrete Logarithms in Elliptic Curve}\label{elliptic_discrete_log_subsection}
In this subsection, we devise an algorithm for the \textsc{Isomorphism Problem} that involves discrete logarithm computations in elliptic curves. We begin with a few preparatory lemmata.
\begin{lemma}\label{ldegree_lemma} $P_\lambda \in E(\F_{q^{\ell^a}})$ and $\mathrm{x}(P_\lambda)$ has degree $\ell^a$. 
\end{lemma}
\begin{proof}
Let $c$ be the smallest positive integer such that $\sigma_E^cP_\lambda = P_\lambda$. To claim $P_\lambda \in E(\F_{q^{\ell^a}})$, it suffices to show $c=\ell^a$. Further, $c=\ell^a$ would also imply that $\mathrm{x}(P_\lambda)$ has degree $\ell^a$, for if $\mathrm{x}(P_\lambda)$ were in a proper subfield of $\F_{q^{\ell^a}}$ then $c$ has to be a proper divisor of $\ell^a$. Since $\sigma_E(P_\lambda) = \lambda P_\lambda$ and $P_\lambda$ has order $\ell^{a+1}$, $c$ equals the order of $\lambda \mod \ell^{a+1}$ in $(\Z/\ell^{a+1}\Z)^\times$. For $\lambda^c = (1+\gamma\ell)^c= 1 \mod \ell^{a+1}$ to hold, it is necessary and sufficient that $\ell^a$ divides $c\gamma$. Hence $c=\ell^a$. 
\end{proof}
\begin{lemma}\label{eigenspace_lemma}
$E(\F_{q^{\ell^a}})[\ell^{a+1}] = \langle P_\lambda\rangle$
\end{lemma}
\begin{proof}
From Lemma \ref{ldegree_lemma}, $\langle P_\lambda \rangle \subseteq E(\F_{q^{\ell^a}})$. Since $E[\ell^{a+1}] = \langle P_\lambda\rangle \oplus \langle P_\mu \rangle$, to claim the lemma it suffices to prove $E(\F_{q^{\ell^a}}) \cap \langle P_\mu \rangle = \{\mathcal{O}\}$. If $E(\F_{q^{\ell^a}}) \cap \langle P_\mu \rangle \neq \{\mathcal{O}\}$, then $\exists P\in E(\F_{q^{\ell^a}}) \cap \langle P_\mu \rangle$ of order $\ell$. Since $P \in E(\F_{q^{\ell^a}})$, $\sigma_E^{\ell^a}P=P$ and since $P \in \langle P_\mu \rangle$, $\sigma_EP =\mu P$. Hence $\mu^{\ell^a}P=P$. Since $P$ has order $\ell$, $\mu^{\ell^a}-1=0 \mod \ell$. Since $\ell$ is a prime, raising to $\ell^{th}$ powers modulo $\ell$ is the identity map implying $\mu =1 \mod \ell$. Since $\gcd(\ell,q-1)= 1$,  this contradicts the fact that $\mu = q \mod \ell$. Thus $E(\F_{q^{\ell^a}}) \cap \langle P_\mu \rangle = \{\mathcal{O}\}$.
\end{proof}
%The group $\Sigma := \langle\sigma_E|\sigma_E^{\ell^a}=1\rangle$ acts on $E(\F_{q^{\ell^a}})$. For $P \in E(\F_{q^{\ell^a}})$, let $\Sigma.P : = \{P,\sigma_EP,\ldots,\sigma_E^{\ell^a-1}P\}$ denote the orbit of $P$ under $\Sigma$.
The group $\Sigma := \langle\sigma_E|\sigma_E^{\ell^a}=1\rangle$ acts on $E(\F_{q^{\ell^a}})$. For $P \in E(\F_{q^{\ell^a}})$, denote the orbit $\{P,\sigma_EP,\ldots,\sigma_E^{\ell^a-1}P\}$ of $P$ under $\Sigma$ as $\Sigma.P$.
\begin{lemma}\label{orbit_lemma} The set $\langle P_\lambda \rangle \setminus \langle \ell P_\lambda\rangle$ is the following disjoint union of orbits $$\langle P_\lambda \rangle \setminus \langle \ell P_\lambda\rangle =\bigcup_{z=1}^{\ell-1} \Sigma.zP_\lambda.$$ 
\end{lemma}
\begin{proof}
For every $z\in \{0,1,\ldots,\ell-1\}$, $zP_\lambda \subseteq \langle P_\lambda\rangle \setminus \langle \ell P_\lambda\rangle$ and $|\Sigma.zP_\lambda| = \ell^a$. Further,
$|\langle P_\lambda\rangle \setminus \langle \ell P_\lambda\rangle| = \ell^{a}(\ell-1)$. It is thus sufficient to prove for distinct $z_1,z_2 \in \{1,2,\ldots,\ell-1\}$ that $z_1P_\lambda \cap z_2 P_\lambda = \emptyset$. If $z_1P_\lambda = \sigma_E^jz_2P_\lambda$ for some $z_1,z_2 \in \{1,2,\ldots,\ell-1\}$ and $j\in\{0,1,\ldots,\ell^a-1\}$ then, $$z_1P_\lambda = \lambda^jz_2P_\lambda \Rightarrow z_1-\lambda^jz_2=0 \mod \ell^a \Rightarrow z_1-(1+\gamma \ell)^jz_2=0\mod \ell^a \Rightarrow z_1=z_2 \mod \ell.$$
\end{proof}
Let $Tr_E:E(\F_{q^{\ell^a}}) \longrightarrow E(\F_{q^{\ell^a}})$ denote the trace like map that sends $P$ to $\sum_{j=0}^{\ell^a-1}\sigma_E^j P$. The next lemma states that distinct $\Sigma$ orbits of $\langle P_\lambda \rangle \setminus \langle\ell P_\lambda \rangle$ have distinct images under $Tr_E$.
\begin{lemma}\label{distinct_trace_lemma} For all $P_1,P_2 \in \langle P_\lambda\rangle \setminus \langle \ell P_\lambda\rangle$, $Tr_E(P_1)=Tr_E(P_2)$ if and only if  $\Sigma.P_1 = \Sigma.P_2$.
\end{lemma}
\begin{proof}
If $P_1,P_2 \in \langle P_\lambda\rangle \setminus \langle \ell P_\lambda\rangle$ and $\Sigma.P_1=\Sigma.P_2$, then $\exists j \in \{0,1,\ldots,\ell^{a}-1\}$ such that $P_2=\sigma_E^jP_1$. Hence, $Tr(P_2)=Tr_E(\sigma_E^jP_1)= \sigma_E^j Tr_E(P_1) = Tr_E(P_1)$. 
We next prove the converse, that is, the ``only if" part of the lemma. 
 For every $\alpha \in \F_q$ at most $q^{\ell^a}-1$ elements in $\F_q^{\ell^a}$ have trace (down to $\F_q$) $\alpha$. If $Tr_E(P_\lambda)=\mathcal{O}$, then $$[E(\F_q):Tr_E(E(\F_{q^{\ell^a}}))] \geq \ell \Rightarrow |E(\F_{q^{\ell^a}})| \leq \left(1+\frac{2}{\sqrt{q}}\right)\frac{q^{\ell^a}}{\ell}.$$
This contradicts the Hasse-Weil bound $|E(\F_{q^{\ell^a}})| \geq q^{\ell^a}-2\sqrt{q^{\ell^a}}$. Thus $Tr_E(P_\lambda) \neq \mathcal{O}$.\\ \\
%Since  $$Tr_E(P_\lambda) = \sum_{i=0}^{\ell^a-1}\sigma_E^iP_\lambda =\sum_{i=0}^{\ell^a-1}\lambda^iP_\lambda = \sum_{i=0}^{\ell^a-1} (1+ i \gamma \ell) P_\lambda =  \ell^a P_\lambda + \frac{\gamma \ell^{a+1} (\ell^{a}-1)}{2}P_\lambda = \ell^{a} P_\lambda,$$
%$Tr_E(P_\lambda) \neq \mathcal{O}$. For otherwise, $P_\lambda$ has order $\ell^{a}$ contradicting the fact that $P_\lambda$ has order $\ell^{a+1}$.\\ \\
Let $P_1,P_2 \in \langle P_\lambda\rangle \setminus \langle \ell P_\lambda\rangle$ and $Tr_E(P_1)=Tr_E(P_2)$. By Lemma \ref{orbit_lemma}, there exists $z_1,z_2 \in \{1,2,\ldots,\ell-1\}$ such that $P_1 \in \Sigma.z_1P_\lambda$ and $P_2 \in \Sigma.z_2P_\lambda$. Hence $Tr_E(P_1) = Tr_E(z_1P_\lambda) = z_1 Tr_E(P_\lambda)$. Likewise, $Tr_E(P_2) = z_2 Tr_E(P_\lambda)$. Since $Tr_E(P_1)=Tr_E(P_2)$, $(z_1-z_2)Tr_E(P_\lambda)=\mathcal{O}$. Since $Tr_E(P_\lambda) \in E(\F_q)[\ell]$, $|E(\F_q)[\ell]|=\ell$ and $Tr(P_\lambda) \neq \mathcal{O}$, the order of $Tr_E(P_\lambda)$ is $\ell$. Hence $z_1-z_2=0 \mod \ell$ thereby implying $\Sigma.P_1 = \Sigma.P_2$.
\end{proof}

\begin{algorithm}[H]
\caption{Root Finding Through Elliptic Curve Discrete Logarithms}\label{elliptic_algorithm}
\begin{algorithmic}[1]
\INPUT  Monic irreducibles $g_1(x),g_2(x) \in \F_q[X]$ of degree $\ell^a$ where $\ell \leq \sqrt{q}$ is a prime not dividing $q(q-1)$.
\OUTPUT A root of $g_1(x)$ in $\F_q(\alpha_2)$ where $\alpha_2 \in \overline{\F}_q$ is a root of $g_2(x)$.
\State Find an elliptic curve $E/\F_q$ with $\ell ||E(\F_q)|$ and $\ell^2 \nmid|E(\F_q)|$.
\State %Compute an element in $E(\F_{q^{\ell^a}})[\ell^{a+1}]$ as follows.
\LineComment Construct $\F_{q^{\ell^a}}$ as $\F_q(\alpha_1)$ where $\alpha_1$ is a root of $g_1(x)$.
\LineComment Find a point $P_1 \in E(\F_{q^{\ell^a}})$ of order $\ell^{a+1}$.
\LineComment $x(P_1)$ is obtained as $f_1(\alpha_1)$ for some $f_1(x)\in \F_q[x]$ of degree less than $\ell^a$.
%\LineComment Since $\F_{q^{\ell^a}}$ is constructed as $\F_q(\alpha_1)$, we obtain $x(P_1) = f_1(\alpha_1)$ for some $f_1(x)\in \F_q[x]$ of degree less than $\ell^a$. 
\LineComment Compute $Tr_E(P_1)$.
\State %Compute an element in $E(\F_{q^{\ell^a}})[\ell^{a+1}]$ as follows.
\LineComment Construct $\F_{q^{\ell^a}}$ as $\F_q(\alpha_2)$ where $\alpha_2$ is a root of $g_2(x)$.
\LineComment Find a point $P_2 \in E(\F_{q^{\ell^a}})$ of order $\ell^{a+1}$.
\LineComment Compute $Tr_E(P_2)$.
\State Find the $z \in \{1,\ldots,\ell-1\}$ such that $Tr_E(P_1)=zTr_E(P_2)$ by solving a discrete logarithm problem in the order $\ell$ cyclic group $E(\F_q)[\ell]$.
\State Compute $zP_2$ and obtain $x(zP_2) = f_2(\alpha_2)$ for some $f_2(x)\in \F_q[x]$ of degree less than $\ell^a$. 
\State Apply Lemma \ref{toeplitz_lemma} to the relation $f_1(\alpha_1) \sim f_2(\alpha_2)$ and output a root of $g_1(x)$ in $\F_q(\alpha_2)$.
\end{algorithmic}
\end{algorithm}
We first argue that algorithm \ref{elliptic_algorithm} is correct. An elliptic curve $E/\F_q$ as required in Step $1$ exists as $\ell^2 \leq \sqrt{q}$ implies $\ell$ has a multiple not divisible by $\ell^2$ in the Hasse interval. As $P_1$ and $P_2$ are both in $E(\F_{q^{\ell^a}})$ and of order $\ell^{a+1}$, by Lemma \ref{eigenspace_lemma}, $P_1,P_2 \in \langle P_\lambda\rangle \setminus \ell\langle P_\lambda\rangle$. Hence by Lemma \ref{orbit_lemma}, there exists $z \in \{1,2,\ldots,\ell-1\}$ such that 
\begin{equation}\label{dlog_orbit_equation}
P_1 = \Sigma.z P_2.
\end{equation}
By Lemma \ref{distinct_trace_lemma}, equation \ref{dlog_orbit_equation} holds if and only if 
\begin{equation}\label{dlog_equation}
Tr_E(P_1) = z Tr_E(P_2). 
\end{equation}
Hence $z$ as desired in Step $4$ exists and further for such a $z$, there exists an integer $j$ such that $P_1 = \sigma_E^j(zP_2)$ implying $f_1(\alpha_1) \sim f_2(\alpha_2)$.\\ \\
The bottleneck in the algorithm happens to be computing a point of order $\ell^{a+1}$ in Steps $2$ and $3$. An algorithm for this task is presented in the subsequent subsection. For now, we discuss the implementation of the other steps. In Step $1$, we generate elliptic curves $E/\F_q$ by choosing a Weierstrass model over $\F_q$ uniformly at random. Then we compute $|E(\F_q)|$ using Schoof's point counting algorithm in $\widetilde{O}(\log^5 q)$ time and check if $\ell||E(\F_q)|$ and $\ell^2 \nmid |E(\F_q)|$. Since $5\ell^3 \leq \sqrt{q}$, the probability that $\ell||E(\F_q)|$ and $\ell^2 \nmid |E(\F_q)|$ is close to $1/(\ell-1)$ \cite[Thm 1.1]{how}. Hence Step $1$ can be completed in time $O(\ell \log^5 q)$. The iterated Frobenius algorithm of von zur Gathen and Shoup \cite{gs} implemented using fast modular composition \cite{ku} computes traces in finite field extensions in nearly linear time. With minor modifications (performing elliptic curve addition in place of finite field addition), it computes $Tr_E(P_1)$ and $Tr_E(P_2)$ in Steps $2$ and $3$ in $\ell^{a+o(1)}\log^{1+o(1)} q$ time. 
%Computing $Tr_E(P_1)$ and $Tr_E(P_2)$ in Steps $2$ and $3$ amounts to computing the trace down to $\F_q$ coordinate wise. An algorithm of Shoup and von zur Gathen \cite{gs} implemented using fast modular composition \cite{ku} computes these traces in $\widetilde{O}(\ell^{a}\log^2 q)$ time.(not quite, have to include elliptic curve addition)
The discrete logarithm computation in Step $4$ can be performed with $\mathcal{O}(\sqrt{\ell})$ $E(\F_q)$-additions by the baby step giant step algorithm. Since $z<\ell$, Step $5$ only takes $O(\log(\ell))$ $E(\F_{q^{\ell^a}})$ additions. From Lemma \ref{toeplitz_lemma}, Step $6$ runs in $\ell^{a+o(1)}\log^{1+o(1)} q$ time. 
\subsection{Lang's theorem and Finding $\ell$ Power Torsion with $\ell$ Isogenies}\label{lang_subsection}
 In \S~\ref{kummer_section} and \S~\ref{artin_schreier_section}, we exploited certain idempotents in proofs of Hilbert's theorem 90 to solve the \textsc{Isomorphism Problem} restricted to the case where $n$ is a power of a prime $\ell$ dividing $p(q-1)$ in linear time. The bottleneck in Algorithm \ref{elliptic_algorithm} for the case $\ell \nmid p(q-1)$ is 
\begin{problem}\label{problem_find_generator} Given a monic irreducible $g(x) \in \F_q[x]$ of prime power $\ell^a$ degree (where $\ell \nmid p(q-1)$), an elliptic curve $E/\F_q$ (where $\ell$ divides $|E(\F_q)|$ but $\ell^2$ does not) and $|E(\F_q)|$, find a generator of $E(\F_{q^{\ell^a}})[\ell^{a+1}]$ where $\F_{q^{\ell^a}}$ is constructed as $\F_q(\alpha)$ for some root $\alpha$ of $g(x)$.
\end{problem}
We next solve Problem \ref{problem_find_generator} using elliptic curve isogenies.
\begin{algorithm}[H]
\caption{Finding $\ell$ Power Torsion:}\label{ltorsion_algorithm}
\begin{algorithmic}[1]
\INPUT 
\LineComment Monic irreducible $g(x) \in \F_q[X]$ of degree $\ell^a$ where $\ell \leq \sqrt{q}$ is a prime not dividing $q(q-1)$. \LineComment An elliptic curve $E/\F_q$ such that $\ell ||E(\F_q)|$ and $\ell^2 \nmid|E(\F_q)|$.
\LineComment $|E(\F_q)|$.
\OUTPUT A point $P \in E$ of order $\ell^{a+1}$ with coordinates in $\F_q[x]/(g(x))$.
\State Construct $\F_{q^{\ell^a}}$ as $\F_q(\alpha)$ for a root $\alpha$ of $g(x)$.
\State Let $\iota:E \longrightarrow \widetilde{E}$ be the isogeny with kernel $\ker(\iota) = E(\F_q)[\ell]$.
\LineComment Compute a Weierstrass equation for $\widetilde{E}/\F_q$.
\LineComment Compute $\phi_\iota(x),\psi_\iota(x) \in \F_q[x]$  such that $\mathrm{x}(\iota(R))=\psi_\iota(\mathrm{x}(R))/\phi_\iota(\mathrm{x}(R)), \forall R \in E$.
\State If $a = 1$
\LineComment Find a point $\widetilde{T} \in \widetilde{E}(\F_q)$ of order $\ell$. 
\LineComment Find a root $\gamma \in \F_q(\alpha)$ of $\phi_\iota(x)- \mathrm{x}(\widetilde{T})\psi_\iota(x) \in \F_q[x]$.
\LineComment Output a point in $E$ with $\mathrm{x}$-coordinate $\gamma$.
\State If $a \neq 1$
\LineComment Find $\widetilde{\alpha}\in \F_q(\alpha)$ of degree $\ell^{a-1}$ and its minimal polynomial $M(x) \in \F_q[x]$ by Lemma \ref{subfield_projection_lemma}.
\LineComment Recursively find a point $\widetilde{P} \in \widetilde{E}(\F_{q^{\ell^{a-1}}})$ of order $\ell^a$ by calling this very algorithm with input $(M(x),\widetilde{E}/\F_q,|E(\F_q)|)$.
\LineComment Find a root $\eta \in \F_q(\alpha)$ of $\phi_\iota(x)- \mathrm{x}(\widetilde{P})\psi_\iota(x) \in \F_q(\widetilde{\alpha})[x]$.%, where $\widetilde{\mathrm{x}}$ is the $\mathrm{x}$-coordinate function on $\widetilde{E}$.
\LineComment Output a point in $E$ with $\mathrm{x}$-coordinate $\gamma$.
\end{algorithmic}
\end{algorithm}
In Step 2, the Weierstrass equation for $\widetilde{E}$ and the polynomials $\psi_\iota(x)$ and $\phi_\iota(x)$ can all be computed in $\ell^{a+o(1)} \log^{2+o(1)} q$ time \cite{cl}. In Step 3, a point $\widetilde{T} \in \widetilde{E}(\F_q)$ of order $\ell$ can be found in $\widetilde{O}(\ell \log q)$ time as follows: generate $\widetilde{R} \in \widetilde{E}(F_q)$ at random and output $\widetilde{T} = |\widetilde{E}(\F_q)|/\ell$ if its not the identity. Note we know don't have to compute $|\widetilde{E}(\F_q)|$ since $|\widetilde{E}(\F_q)|=|E(\F_q)|$. The root finding in Step $3$ takes $\ell^{2+o(1)} \log^{2+o(1)} q$ time using \cite{ks,ku}. By \cite{cl}, a root $\gamma$ of $\phi_\iota(x)- \mathrm{x}(\widetilde{T})\psi_\iota(x) \in \F_q[x]$ has degree $\ell$ and the two points in $E$ with $\mathrm{x}$-coordinate $\gamma$ both have order $\ell^2$ and are in $E(\F_{q^\ell})$. Thus the output at the end of Step $3$ is correct. Likewise, in Step $4$, by \cite{cl}, a root $\eta$ of $\phi_\iota(x)- \mathrm{x}(\widetilde{P})\psi_\iota(x)$ has degree $\ell^a$ and the two points in $E$ with $\mathrm{x}$-coordinate $\eta$ both have order $\ell^{a+1}$ and are in $E(\F_{q^{\ell^a}})$. Hence the output at the end of Step $4$ is correct. The root finding in Step $4$ takes $\ell^{a+1+o(1)} \log^{1+o(1)} q$ time using \cite{ks,ku} and is the bottleneck. The number of recursive calls is at most $a$ which being logarithmic in $\ell^a$ can be ignored in the run time analysis.\\ \\
Using Algorithm \ref{ltorsion_algorithm} as a subroutine, Algorithm \ref{elliptic_algorithm} solves the \textsc{Isomorphism Problem} restricted to the special case when $n=\ell^a$ for some prime $\ell$ such that $\ell \nmid q(q-1)$ and $5\ell^3 \leq \sqrt{q}$. The restriction $5\ell^3 \leq \sqrt{q}$ may be removed without loss of generality. For if $5\ell^3 > \sqrt{q}$ in the \textsc{Isomorphism Problem}, we may pose the problem over a small degree extension $\F_{q^d}$ instead of $\F_q$ where $d$ is the smallest positive integer such that $\ell \leq \sqrt{q^d}$ and $\ell \nmid d$ (c.f.\cite{rai}). In summary, we have proven
\begin{theorem}\label{elliptic_theorem}
Algorithm \ref{elliptic_algorithm} solves the \textsc{Isomorphism Problem} restricted to the special case when $n=\ell^a$ for some prime $\ell \nmid q(q-1)$ in $\ell^{a+1+o(1)} \log^{1+o(1)} q + O(\ell \log^5 q)$ time. 
\end{theorem}
The running time is subquadratic in the input degree $\ell^a$ if $a>1$. If $a=1$, that is, if the input degree is a prime $\ell$, the running time is quadratic. The question if a sub quadratic algorithm for the later case exists remains open. We look to Lang's theorem, an elliptic curve analogue of Hilbert's theorem 90 in hopes of solving the bottleneck Problem \ref{problem_find_generator} in subquadratic time. Lang's theorem states that the first cohomology group $H^1(\F_q,E)$ of an elliptic curve $E$ over $\F_q$ is trivial. That is, the Lang map $\psi: E \longrightarrow E$ taking $P$ to $\sigma_E(P)-P$ 
%$$\psi: E \longrightarrow E$$
%$$\ \ \ \ \ \ \ \ \ \ \ \ \ \ \ \ P \longmapsto\sigma_E(P)-P$$
is surjective. Problem \ref{problem_find_generator} is rephrased in terms of computing preimages under the Lang map as the following Problem \ref{problem_lang}. Problems \ref{problem_find_generator} and \ref{problem_lang} are equivalent since the preimage of $E(\F_q)[\ell] \setminus \{O\}$ under $\psi$ is $E(\F_{q^{\ell^a}})[\ell^{a+1}]$.
\begin{problem}\label{problem_lang}
Given a monic irreducible $g(x) \in \F_q[x]$ of prime power $\ell^a$ degree (where $\ell \nmid p(q-1)$), an elliptic curve $E/\F_q$ (where $\ell$ divides $|E(\F_q)|$ but $\ell^2$ does not) and $|E(\F_q)|$, find a preimage under the Lang map $\psi$ of $E(\F_q)[\ell] \setminus \{O\}$ in $E(\F_{q^{\ell^a}})$ where $\F_{q^{\ell^a}}$ is constructed as $\F_q(\alpha)$ for some root $\alpha$ of $g(x)$.
\end{problem}
\textsc{Open Problem:} \textit{Solve Problem \ref{problem_find_generator} or Problem \ref{problem_lang} in time sub quadratic in $\ell^a$.}
 %Using algorithm \ref{ltorsion_algorithm} (to solve problem \ref{problem_find_generator}) as a subroutine in algorithm \ref{elliptic_algorithm} yields the following.

%\section*{Acknowledgement}
%I thank Matthias Flach, Ming-Deh Huang, Eric Rains and Chris Umans for valuable discussions.

\bibliographystyle{plain}

\begin{thebibliography}{99}% Replace 9 by 99 if 10 or more references
%
% Please note the use of "\and" between author names below
%
\bibitem[All02]{all} B. Allombert, Explicit Computation of Isomorphisms between Finite Fields, \emph{Finite Fields and Their Applications 8, 332–342 (2002)}


\bibitem[BGJT14]{bgjt} R. Barbulescu, P. Gaudry, A. Joux \and E. Thome, A Heuristic Quasi-Polynomial Algorithm for Discrete Logarithm in Finite Fields of Small Characteristic, \emph{ Advances in Cryptology – EUROCRYPT 2014}, P 1-16.

\bibitem[Ben81]{ben} M. Ben-Or, Probabilistic algorithms in Finite Fields,\emph{FOCS} (1981), pp. 394-398.

\bibitem[Ber67]{ber}E. R. Berlekamp, Factoring Polynomials Over Finite Fields, \emph{Bell System Tech. J.}, 46:1853-1849. 1967.

\bibitem[BDDFS17]{bddfs} L. Brieulle, L. De Feo, J. Doliskani,
J-P. Flori and E. Schost, Computing isomorphisms and embeddings of finite fields, \url{https://arxiv.org/abs/1705.01221}. To appear in Mathematics of Computation, \url{https://doi.org/10.1090/mcom/3363}


\bibitem[CZ81]{cz} D. G. Cantor \and H. Zassenhaus, A new algorithm for factoring polynomials over finite fields, \emph{Math. Comp.}, vol. 36, 587-592, 1981.

\bibitem[CL13]{cl} J-M Couveignes \and R Lercier, Fast construction of irreducible polynomials over finite fields,\emph{ Israel Journal of Mathematics}, The Hebrew University Magnes Press, 2013, 194(1), pp.77-105.

\bibitem[vzGS92]{gs} J. von zur Gathen \and V. Shoup, Computing Frobenius maps and factoring polynomials, \emph{Comput. Complexity}, vol. 2, 187-224, 1992.

\bibitem[Gra08]{gra} A. Granville, Smooth numbers: computational number theory and beyond, \emph{Algorithmic Number Theory}, MSRI Publications Volume 44, 2008.


\bibitem[How]{how} E. Howe, On the group orders of elliptic curves over finite fields, \emph{Compositio Mathematica} (1993): 229-247.


%\bibitem[KL94]{kl}{\bibname E. Kaltofen \and A. Lobo}, Factoring high-degree polynomials by the black box Berlekamp algorithm,\emph{ ISSAC '94 Proceedings of the international symposium on Symbolic and algebraic computation}, Pages 90 - 98.

%\bibitem[KS98]{ks}{\bibname E. Kaltofen \and V. Shoup}, Subquadratic-time factoring of polynomials over finite fields, \emph{Math. Comput.}, 67(223):1179-1197, July 1998. 

\bibitem[KS99]{ks} E. Kaltofen \and V. Shoup, Fast polynomial factorization over high algebraic extensions of finite fields. \emph{In Proc. 1997 Internat. Symp. Symbolic Algebraic Comput. (ISSAC'97) }, pages 184-188.
 
\bibitem[KU08]{ku} K. Kedlaya \and C. Umans, Fast modular composition in any characteristic, \emph{FOCS}: 2008, pages 146-155.

\bibitem[Lan78]{lan} S. Lang, Algebraic groups over finite fields, \emph{American Journal of Mathematics 78}: 555-563.

\bibitem[Len87]{len} H. W. Lenstra Jr, Factoring integers with elliptic curves,\emph{ Annals of Mathematics} 126 (3): 649-673. (1987).

\bibitem[Moo1889]{moo}  E. H. Moore, A doubly-infinite system of simple groups, \emph{Bull. New York Math. Soc. 3 (1893),
73-78; Math. Papers read at the Congress of Mathematics (Chicago, 1893), Chicago, 1896}, pp. 208-242. 

\bibitem[Nar2016]{nar} A. K. Narayanan,  Fast computation of isomorphisms between finite fields using elliptic curves, \url{https://arxiv.org/abs/1604.03072}

\bibitem[Pin92]{pin} R. G. E. Pinch, Recognizing elements of finite fields,\emph{Cryptography and Coding II}, pages 193-197, 1992. 

\bibitem[Rai08]{rai} E. Rains, Efficient computation of isomorphisms between finite fields. 

\bibitem[Sch95]{sch} R. Schoof, Counting Points on Elliptic Curves over Finite Fields, \emph{J. Theor. Nombres Bordeaux 7}:219-254, 1995.

\bibitem[Sho99]{sho_mp} V. Shoup, Efficient computation of minimal polynomials in algebraic extensions of finite fields,\emph{ ISSAC '99}, Pages 53-58.

\bibitem[Sho95]{sho_ir} V. Shoup, Fast construction of irreducible polynomials over finite fields", \emph{Journal of Symbolic Computation}, 1994.

\bibitem[Zie74]{zie} N. Zierler, A conversion algorithm for logarithms on $GF(2^n)$, \emph{Journal of Pure and Applied Algebra}, 4:353-356, 1974.

\end{thebibliography}

\end{document}